\documentclass[letterpaper, 10 pt, journal, twoside]{IEEEtran}

\usepackage{amsmath,amssymb,mathtools,amsthm}
\usepackage{hyperref}
\usepackage{bm}
\usepackage{parskip}
\usepackage{subfig}
\usepackage{nicefrac}
\usepackage{algorithm}
\usepackage{setspace}
\usepackage[noEnd=true,commentColor=green!50!black,beginComment=// ,beginLComment=/* ,endLComment= */]{algpseudocodex}

\title{Efficient streaming dynamic mode decomposition}
\author{Aditya Kale, \IEEEmembership{Student Member, IEEE,} Marcos Netto, \IEEEmembership{Senior Member, IEEE,} Xinyang Zhou, \IEEEmembership{Member, IEEE}
\thanks{This work was authored in part by the National Renewable Energy Laboratory, operated by Alliance for Sustainable Energy, LLC, for the U.S. Department of Energy (DOE) under contract no. DE-AC36-08GO28308. The views expressed in the article do not necessarily represent the views of the DOE or the U.S. Government. The U.S. Government and the publisher, by accepting the article for publication, acknowledge that the U.S. Government retains a nonexclusive, paid-up, irrevocable, worldwide license to publish or reproduce the published form of this work, or allow others to do so, for U.S. Government purposes.}
\thanks{This work was partially supported by the Laboratory Directed Research and Development program at the National Renewable Energy Laboratory and by the National Science Foundation under Grant 2328241.}
\thanks{A. Kale and M. Netto are with the Department of Electrical and Computer Engineering, New Jersey Institute of Technology, Newark, NJ 07102, USA. X. Zhou is with the Power Systems Engineering Center, National Renewable Energy Laboratory, Golden, CO 80401, USA.}
}

\newcommand{\norm}[1]{\left\lVert#1\right\rVert}
\newcommand{\xk}{\bm{x}_{k}}
\newcommand{\yk}{\bm{y}_{k}}

\newcommand{\xkt}{\tilde{\bm{x}}_{k}}
\newcommand{\ykt}{\tilde{\bm{y}}_{k}}

\newcommand{\xtm}{\tilde{\bm{x}}_}
\newcommand{\ytm}{\tilde{\bm{y}}_}
\newcommand{\St}{\tilde{\bm{A}}}
\newcommand{\Skt}{\tilde{\bm{A}}_{k}}

\newcommand{\Qk}{\bm{Q}_{k}}

\newcommand{\Qm}{\bm{Q}_}
\newcommand{\T}{^\top}

\newcommand{\Xkt}{\tilde{\bm{X}}_{k}}
\newcommand{\Ykt}{\tilde{\bm{Y}}_{k}}

\newcommand{\A}{\bm{A}}
\newcommand{\Xk}{\bm{X}_{k}}
\newcommand{\Yk}{\bm{Y}_{k}}
\newcommand{\Xm}{\bm{X}_}
\newcommand{\Ym}{\bm{Y}_}

\newcommand{\xm}{\bm{x}_}
\newcommand{\ym}{\bm{y}_}
\newcommand{\C}{\bm{C}}

\newcommand{\GXm}[1]{\bm{G}_{\bm{X}_{#1}}}
\newcommand{\GYm}[1]{\bm{G}_{\bm{Y}_{#1}}}
\newcommand{\Cm}{\bm{C}_}
\newcommand{\Ck}{\bm{C}_{k}}
\newcommand{\GXk}{\bm{G}_{\Xk}}
\newcommand{\GYk}{\bm{G}_{\Yk}}

\newtheorem{theorem}{Theorem}
\newtheorem{remark}{Remark}

\begin{document}

\maketitle
\thispagestyle{empty}
\pagestyle{empty}

\begin{abstract}
We propose a reformulation of the streaming dynamic mode decomposition method that requires maintaining a single orthonormal basis, thereby reducing computational redundancy. The proposed efficient streaming dynamic mode decomposition method results in a constant-factor reduction in computational complexity and memory storage requirements. Numerical experiments on representative canonical dynamical systems show that the enhanced computational efficiency does not compromise the accuracy of the proposed method.

\end{abstract}


\section{Introduction} \label{sec:intro}
\IEEEPARstart{D}{ynamic} mode decomposition (DMD) has been widely used to reveal the discrete spectrum inherent in complex dynamical systems \cite{Kutz2016}. Applications range from epidemiology \cite{Proctor2015} and neuroscience \cite{Brunton2016} to power and energy systems \cite{Netto2021} and control \cite{Proctor2016}, showcasing its prominence. Over the years, various algorithms have been proposed to enhance or extend the original DMD method \cite{Schmid2010}, some of which are referenced here. For a more comprehensive review, interested readers are directed to \cite{Colbrook2023}. 

The original DMD method \cite{Schmid2010} and most of its variations have been developed for \emph{batch processing}---that is, in plain terms, to collect data first and then analyze it. Conversely, DMD methods tailored to \emph{stream processing}---in which case data streams are sequentially analyzed as they arrive, one snapshot at a time---have been relatively less explored. Unlike batch-processing DMD algorithms that require access to the full dataset upfront, stream-processing algorithms must output results continuously as new data becomes available. 

This shift in computing paradigm, from batch processing to real-time insights in data \cite{vanDongen2020}, requires a reformulation of the batch algorithm. The streaming dynamic mode decomposition (sDMD) method \cite{Hemati2014} offers such reformulation, introducing fundamental structural changes to support stream processing. The subsequent application of sDMD to characterize and distill the dynamics of a wind farm \cite{Liew2022} exemplifies its potential for real-time analysis of energy systems and beyond. 

A few select works build upon the original sDMD \cite{Hemati2014} method, which, like its batch-processing counterpart algorithm \cite{Schmid2010}, invokes a standard least-squares solution \cite{Abolmasoumi2022}. A straightforward modification that replaces a least-squares solution with a \emph{total least-squares} \cite{Hemati2017} solution enhances the sDMD performance in extracting interpretable information from noisy data, as demonstrated in \cite{Hemati2016} using experimental time-resolved particle image velocimetry measurements. 

Leveraging incremental singular value decomposition \cite{Brand2002}---a method for updating a matrix's prior singular value decomposition as new data arrives, rather than computing it from scratch each time---represents another advancement in sDMD. Recall that the original DMD \cite{Schmid2010} and sDMD \cite{Hemati2014} both utilize singular value decomposition as part of their algorithm. Not only does the incremental singular value decomposition fit naturally within the stream processing paradigm, but it also lends itself to sparsification \cite{Berry1992} after the incremental process, as demonstrated in \cite{Matsumoto2017}. 

Yet another advancement to the class of stream-processing DMD methods \cite{Zhang2019, Susuki2022} leverages the \emph{matrix inversion lemma}, also known as the Sherman-Morrison-Woodbury formula or Woodbury matrix identity, to compute the inverse of a rank-$n$ update to a matrix whose inverse has previously been calculated. This development enabled the introduction of a weighting scheme to emphasize recent data and a windowed approach, effectively adapting the algorithm to evolving system dynamics \cite{Zhang2019}. Building upon \cite{Matsumoto2017} and \cite{Zhang2019}, the work in \cite{Alfatlawi2020} combines incremental singular value decomposition with the use of the matrix inversion lemma, thereby extending its applicability to dynamical systems with control inputs. A formulation tailored explicitly to low-rank datasets also exists \cite{Nedzhibov2023a}, designed for applications where the number of snapshots is less than the number of state variables.

In this letter, we are particularly interested in improving the computational efficiency of sDMD. We observe that the core functionality of sDMD is based on expanding and rotating two separate orthonormal bases corresponding to the column spaces of the same underlying dynamical system---an approach that introduces redundancy.

Motivated by this insight, the primary contribution of this work is to demonstrate that maintaining only a single orthonormal basis is sufficient to characterize the system dynamics accurately. This modification eliminates the need for dual basis updates, resulting in a constant-factor reduction in computational complexity compared to sDMD. We refer to the proposed method as efficient streaming dynamic mode decomposition (esDMD), as in computer science, efficient computation refers to designing algorithms and systems that perform calculations with minimal time and memory usage.

This letter proceeds as follows. Section \ref{sec:prelim} briefly introduces the derivation of the linear operators underlying DMD \cite{Schmid2010} and sDMD \cite{Hemati2014}. Section \ref{sec:esdmd} presents our main contribution, detailing modifications to the sDMD formulation. Section \ref{sec:exp} compares the proposed esDMD with sDMD and benchmarks these methods against the batch-processing DMD, using canonical systems as testbeds. Section \ref{sec:conc} concludes this letter and outlines potential directions for future research.

\subsection{Notation}
Lowercase (uppercase) boldface letters, e.g., $\bm{x}$ ($\bm{X}$), denote vectors (matrices). Vector and matrix slicing follow Python notation; for instance, \(\bm{X}[:i,\,:j]\) refers to the submatrix of \(\bm{X}\) consisting of its first \(i\) rows and first \(j\) columns. \(\bm{X}^\top\), \(\bm{X}^+\), \(\norm{\mathbf{X}}_F\), \(\mathrm{cols}(\bm{X})\), and \(\mathrm{span}(\bm{X})\) denote, respectively, the transpose, the Moore-Penrose pseudoinverse, the Frobenius norm, the number of columns, and the column space of \(\bm{X}\). $\bm{I}$ is an identity matrix of appropriate dimensions, dictated by the context in which it is used. The symbol $\equiv$ denotes equivalence. The \(\ell^2\)-norm of a vector \(\bm{x}\) is \(\norm{\bm{x}}\). A vector \(\bm{x}^\parallel\) is the projection of \(\bm{x}\) onto some orthonormal basis.

\section{Preliminaries} \label{sec:prelim}

\subsection{Batch-processing dynamic mode decomposition (DMD)}
Let an autonomous dynamical system evolving on a finite, $n$-dimensional manifold $\mathbb{X}$ be:
\begin{equation}
    \bm{x}_{i} = \bm{f}(\bm{x}_{i-1}), \; \quad\text{for discrete time } i\in\mathbb{Z},
    \label{eq:dynsys_map}
\end{equation}
where $\bm{f}:\mathbb{X} \to \mathbb{X} \subset \mathbb{R}^{n}$ is a nonlinear, vector-valued map. A set of states $\{\xm{i}\}_{i=1}^k$, $k \ge n$, evolving according to \eqref{eq:dynsys_map} can be used to build the following data matrices:
\begin{equation}
\Xk = \begin{bmatrix}
    \xm1&\dots&\xm{k}
\end{bmatrix},\quad
\Yk = \begin{bmatrix}
    \ym1&\dots&\ym{k}
\end{bmatrix},
\label{eq:datamatrices}
\end{equation}
where $\bm{y}_{i}=\bm{x}_{i+1}$ and every tuple $(\xm{i},\ym{i})$ is called a snapshot pair. The batch-processing DMD computes a linear operator $\A_k \in \mathbb{R}^{n\times n}$ that minimizes, for all snapshot pairs, the cost
\begin{align}
    J = \sum_{i=1}^k \norm{\ym{i} - \A_k\,\xm{i}}^2
    = \norm{\Yk - \A_k\,\Xk}_F^2. \label{eq:dmdcost2}
\end{align}

If the tuples $(\xm{i},\ym{i})$ in \eqref{eq:datamatrices} follow a Gaussian distribution,
\begin{equation}
    \A_k = \Yk\,\Xk^+
    \label{eq:dmdminimizer}
\end{equation}
is the optimal (maximum likelihood) solution to \eqref{eq:dmdcost2} \cite{Abolmasoumi2022}.

\subsection{Projected, lower-dimensional operator}
Let $\Xk = \bm{U}_{x}\,\bm{\Sigma}_{x}\,\bm{V}_{x}\T$ be the singular value decomposition of $\Xk\in\mathbb{R}^{n\times k}$; and for an integer $r < n$, let $\Qm{\Xk}\in\mathbb{R}^{n\times r}$ be an orthonormal basis for the column space of $\Xk$ obtained by retaining the first $r$ columns of $\bm{U}_{x}$. By projecting the operator $\A_k$ \eqref{eq:dmdminimizer} onto this basis as
\begin{align}
   \St_k &= \Qm{\Xk}\T \A_k\,\Qm{\Xk} = \Qm{\Xk}\T\,\Yk\,\Xk^+\Qm{\Xk},
   \label{eq:reducedA}
\end{align}
the eigenvalues and eigenvectors of $\A_k\in\mathbb{R}^{n\times n}$ can be computed from those of a much smaller matrix $\St_k\in\mathbb{R}^{r\times r}$. Note that the amount of storage required to compute $\St_k$ increases as the number of snapshot pairs, $k$, grows.

\subsection{Streaming dynamic mode decomposition (sDMD)}
The batch-processing DMD can be reformulated to handle streaming data, where snapshot pairs $(\xm{i}, \ym{i})$ are acquired incrementally, one at a time \cite{Hemati2014}. Two orthonormal bases, $\Qm{\Xk}$ and $\Qm{\Yk}$, are maintained to span the evolving column spaces of $\Xk$ and $\Yk$, respectively. These bases are computed and updated incrementally as new snapshot pairs arrive. 

Define $\Xkt \coloneqq \Qm{\Xk}\T \Xk$ and $\Ykt \coloneqq \Qm{\Yk}\T \Yk$. Since $\Xk$ has $n$ linearly independent rows, $\Xk^+ = \Xk\T (\Xk\Xk\T)^{-1}$, and the reduced-order linear operator in \eqref{eq:reducedA} can be expressed as:
\begin{equation*}
    \St_k = \Qm{\Xk}\T\,\Qm{\Yk}\,\Ck\,\GXk^+,
\end{equation*}
where $\Ck = \Ykt \Xkt\T$ and $\GXk = \Xkt \Xkt\T$ are $r \times r$ matrices representing accumulated inner products in reduced coordinates. Since only a single snapshot pair $(\xm{i}, \ym{i})$ is available at each time step in stream processing, $\Ck$ and $\GXk$ must be updated incrementally using projected coordinate vectors $\xkt = \Qm{\Xk}\T \xm{i}$ and $\ykt = \Qm{\Yk}\T \ym{i}$. These updates follow:
\begin{align}
    \Ck &= \C_{k-1} + \ytm{k}\,\xtm{k}\T, \label{eq:Ckupdate} \\
    \GXk &= \bm{G}_{\Xm{k-1}} + \xtm{k}\,\xtm{k}\T, \label{eq:Gxupdate}
\end{align}
where $\C_{k-1} = \sum_{i=1}^{k-1} \ytm{i} \xtm{i}\T$ and
$\bm{G}_{\Xm{k-1}} = \sum_{i=1}^{k-1} \xtm{i} \xtm{i}\T$ are accumulated terms from prior time steps.

At each time step, the bases $\Qm{\Xk}$ and $\Qm{\Yk}$ are updated through Gram-Schmidt orthonormalization to incorporate the latest data stream. To maintain a fixed basis rank $r$, a proper orthogonal decomposition compression step is applied whenever the number of basis vectors exceeds $r$. This is accomplished by computing the eigendecompositions of the Gram matrices $\GXk$ and $\GYk$, where $\GYk = \bm{G}_{\Ym{k-1}} + \ytm{k}\,\ytm{k}\T$, and $\bm{G}_{\Ym{k-1}} = \sum_{i=1}^{k-1} \ytm{i} \ytm{i}\T$. 

We are now in the position to present the main contribution of this letter.

\section{Efficient streaming dynamic mode decomposition (esDMD)} \label{sec:esdmd}
The data matrices in \eqref{eq:datamatrices} admit a compact representation, as their column spaces lie within low-dimensional subspaces of the state space. These subspaces can be efficiently characterized by orthonormal bases---a property that is central to the design of sDMD \cite{Hemati2014}, where two separate bases are maintained for the $\bm{x}$ and $\bm{y}$ snapshots.

However, maintaining two distinct bases, $\Qm{\Xk}$ and $\Qm{\Yk}$, is unnecessary when considering the temporal continuity of the data stream. Since $\ym{i} = \xm{i+1}$, it directly follows that $\xm{i} = \ym{i-1}$. Therefore, each streaming snapshot pair can be equivalently expressed as $(\ym{i-1}, \ym{i})$, enabling the use of a single evolving basis without loss of generality.

\begin{theorem}\label{theorem:single_base_sufficiency}
Consider a discrete-time, autonomous dynamical system evolving on a finite, $n$-dimensional manifold $\mathbb{X}$,
\begin{equation*}
    \bm{x}_{i} = \bm{f}(\bm{x}_{i-1}),
\end{equation*}
where $i\in\mathbb{Z}$ and $\bm{f}:\mathbb{X} \to \mathbb{X} \subset \mathbb{R}^{n}$. Let streaming data pairs $\left(\xm{i},\ym{i}\right) \equiv \left(\ym{i-1},\ym{i}\right)$ form matrices $\Xk=\begin{bmatrix}\xm1&\ldots&\xk\end{bmatrix}$ and $\Yk=\begin{bmatrix}\ym1&\ldots&\yk\end{bmatrix}$.
There exists $\Qk\in\mathbb{R}^{n\times r},\,r\leq n$, 
such that $\mathrm{span}(\Qk)=\mathrm{span}(\Xk)$ and $\mathrm{span}(\Qk)=\mathrm{span}(\Yk)$, for all $i \geq 1$.
\end{theorem}

\begin{proof}
    The proof is by induction. 
    
    First, let $i=1$, and
    \begin{align*}
       \Xm1=\bigl[\xm1\bigr],\quad\Ym1=\bigl[\ym1\bigr]. 
    \end{align*}
    Initialize
    \begin{equation}
    \Qm1=\text{\textsc{qr}}(\begin{bmatrix}
        \xm1 & \ym1
    \end{bmatrix}),
    \label{eq:Qinit}
    \end{equation}
    where $\text{\textsc{qr}}(\cdot)$ returns an orthonormal basis obtained via QR decomposition. This immediately implies $\mathrm{span}(\Qm1)=\mathrm{span}(\begin{bmatrix}\Xm1&\Ym1\end{bmatrix})$. 
    
    Next, assume for $i=k,\,k>1$, $\exists\Qk$, such that $\mathrm{cols}(\Qk)<r$ and $\mathrm{span}(\Qk)=\mathrm{span}(\Xk)$ and $\mathrm{span}(\Qk)=\mathrm{span}(\Yk)$. We now prove the statement for $i=k+1$ using Gram-Schmidt orthonormalization. Let,
    \begin{align*}
        \ym{k+1}^\parallel=\Qk\left(\Qk\T\ym{k+1}\right),\quad\bm{e}_{k+1}=\ym{k+1}-\ym{k+1}^\parallel.
    \end{align*}
    Set
    \begin{align*}
        \Qm{k+1}=\begin{cases}
            \begin{bmatrix}
                \Qk & \bm{p}_{k+1}
            \end{bmatrix} ,& \nicefrac{\bm{e}_{k+1}}{\norm{\ym{k+1}}}>\epsilon\\
            \Qk ,& \nicefrac{\bm{e}_{k+1}}{\norm{\ym{k+1}}}\le\epsilon
        \end{cases},
    \end{align*}
    where, $\bm{p}_{k+1}=\nicefrac{\bm{e}_{k+1}}{\norm{\bm{e}_{k+1}}}$ and $\epsilon$ is some minimum tolerance. Since, it is established that $\xm{k+1}=\ym{k}$, then, by construction,
    \begin{align*}
        \mathrm{span}(\Qm{k+1})&=\mathrm{span}(\begin{bmatrix}
            \Qk & \bm{p}_{k+1}
        \end{bmatrix})\\
        &= \mathrm{span}(\begin{bmatrix}
            \Xk & \xm{k+1}
        \end{bmatrix})=\mathrm{span}(\Xm{k+1})\\
        &= \mathrm{span}(\begin{bmatrix}
            \Yk & \ym{k+1}
        \end{bmatrix})=\mathrm{span}(\Ym{k+1}).
    \end{align*}
    Thus, by induction, Theorem \ref{theorem:single_base_sufficiency} holds for all $i\geq1$ and $\mathrm{cols}(\Qm{i})\leq r$.
\end{proof}

\begin{remark}
After the initial period, if $\mathrm{cols}(\Qm{i})>r$, $\Qm{i}$ is compressed and rotated to align with the dynamics not spanned by the current basis matrix.
\end{remark}

In this work, the initialization in equation \eqref{eq:Qinit} serves as the main pivot point from \cite{Hemati2014}, allowing the use of a single basis instead of two. This initialization is inspired by Algorithm 3 of \emph{Exact DMD} \cite{Tu2014}, where the orthonormal basis is computed for the joint column space of $\begin{bmatrix}
   \Xk & \Yk 
\end{bmatrix}$.

\subsection{Rank-preserving updates}
Similar to streaming DMD, at time step $i=k+1$, we use the eigendecomposition of the Gram matrix $\GYk$ to transform the basis $\Qk$ into $\Qm{k+1}$. Next, we outline the update procedure for the orthonormal basis after its initial construction if the newest snapshot introduces a mode not captured in the current basis:
\begin{enumerate}
    \item Append a new column to the current basis:
    \begin{align*}
        \Qk' = \begin{bmatrix}
            \Qk & \bm{p}_{k+1}
        \end{bmatrix}.
    \end{align*}
    \item Pad the existing Gram matrix with zeros to match the dimensions of the expanded basis:
    \begin{align*}
        \GYk' &= \begin{bmatrix}
            \GYk & \bm{0} \\
            \bm{0} & 0
        \end{bmatrix}.
    \end{align*}
    \item Retain the top $r$ eigenvectors of $\GYk'$:
    \begin{align*}
        \GYk' &= \bm{W}_{k}'\,\bm{\Lambda}_k'\,\bm{W}_k'^{\top},\\
        \tilde{\bm{W}}_k' &= \bm{W}_k'\left[:,\,:r\right].
    \end{align*}
    \item Rotate $\Qk'$ to obtain the updated basis:
    \begin{align*}
        \Qm{k+1} &= \Qk'\,\tilde{\bm{W}}_k'.
    \end{align*}
\end{enumerate}

Note that, by construction, $\Xk$ always lies in $\mathrm{span}(\Qk)$. Moreover, rotation by eigenvectors preserves the span, since $\Qm{k+1}$ inherits the column space of $\Qk$. Furthermore, this rotation also preserves orthonormality, since
\begin{align*}
    \Qm{k+1}\T\Qm{k+1} &= \tilde{\bm{W}}_k'^{\top}\Qk'^{\top}\Qk'\,\tilde{\bm{W}}_k'\\
    &= \tilde{\bm{W}}_k'^{\top}\bm{I}\,\tilde{\bm{W}}_k' = \bm{I}.
\end{align*}

\subsection{Low-dimensional operator}
With a single orthonormal basis, $\Qk$, at any given time step $i=k$, equation \eqref{eq:reducedA} can be written as
\begin{align}
    \St_k &= \Qk\T\Yk\,\Xk^+\Qk \nonumber \\
    &= \Qk\T\Yk\,\Xk\T\left(\Xk\Xk\T\right)^+ \label{eq:singleSt_pinvidentity}\\
    &= \Qk\T\Qk\Ykt\,\Xkt\T\Qk\T\left(\Qk\,\Xkt\,\Xkt\T\Qk\T\right)^+ \label{eq:singleSt_projections} \\
    &= \Ykt\,\Xkt\T\left(\Xkt\,\Xkt\T\right)^+, \label{eq:singleSt_last}
\end{align}
where \eqref{eq:singleSt_pinvidentity} derives from the pseudoinverse identity $\Xk^+=\Xk\T\left(\Xk\,\Xk\T\right)^+$; \eqref{eq:singleSt_projections} uses the low dimensional projections $\Xkt=\Qk\T\Xk$ and $\Ykt=\Qk\T\Yk$; and \eqref{eq:singleSt_last} can be rewritten with data vectors using \eqref{eq:Ckupdate} and \eqref{eq:Gxupdate}, as
$\Skt = \C_k\,\GXk^+$.

\subsection{Computational complexity and redundancy}
Algorithm \ref{alg:esdmd-algo} updates the basis at each time step and stores all relevant intermediate variables between calls, allowing for their reuse in subsequent time steps. The matrices $\GXm{i}$ and $\Cm{i}$ produced by Algorithm \ref{alg:esdmd-algo} are subsequently used to compute the linear operator, as described in Algorithm \ref{alg:linear-operator}.

The primary computational costs in sDMD arise from performing eigendecompositions of the matrices $\GXm{i}$ and $\GYm{i}$, computing the Moore--Penrose pseudoinverse of $\GXm{i}$, and calculating and decomposing the matrix $\St_i$. Taken together, these operations result in an overall computational complexity of $\mathcal{O}(nr^3)$. Our proposed algorithm shares this asymptotic complexity; however, it achieves a constant factor speed-up by eliminating redundant computations.

\begin{algorithm}[!t]
\caption{esDMD: Efficient streaming dynamic mode decomposition}  
\label{alg:esdmd-algo}
\begin{algorithmic}
\Ensure At time step $i=1$, $\texttt{new\_direction}=\texttt{false}$
\Procedure{UpdateBasis}{$(\xm{i},\ym{i})$}
\If{$i=1$}
\State $\Qm{i}\gets\Call{qr}{\begin{bmatrix}
                             \xm{i} & \ym{i}
                             \end{bmatrix}}$
\State $\xtm{i}\gets\Qm{i}\T\xm{i},\quad\ytm{i}\gets\Qm{i}\T\ym{i}$
\State $\GXm{i}\gets\xtm{i}\,\xtm{i}\T,\quad\GYm{i}\gets\ytm{i}\,\ytm{i}\T$
\State $\Cm{i}\gets\ytm{i}\,\xtm{i}\T$
\State $\ytm{-}\gets\ytm{i}$
\State \Return $\Qm{i},\,\GXm{i},\,\Cm{i}$
\EndIf
\State $\ym{i}^\parallel\gets\Qm{i-1}\left(\Qm{i}\T\,\ym{i}\right)$
\State $\bm{e}_i\gets\ym{i}-\ym{i}^\parallel$
\If{$\nicefrac{\norm{\bm{e}_i}}{\norm{\ym{i}}}>\epsilon$}
\State $\bm{p}_i\gets\nicefrac{\bm{e}_i}{\norm{\bm{e}_i}}$
\State $\Qm{i}'\gets\begin{bmatrix}
\Qm{i-1} & \bm{p}_i
\end{bmatrix}$\\[0.05em]
\State $\GXm{i}'\gets\begin{bmatrix}
\GXm{i-1} & \bm{0}\\
\bm{0} & 0
\end{bmatrix},\quad\GYm{i}'\gets\begin{bmatrix}
\GYm{i-1} & \bm{0}\\
\bm{0} & 0
\end{bmatrix}$\\[0.05em]
\State $\Cm{i}'\gets\begin{bmatrix}
\Cm{i-1} & \bm{0}\\
\bm{0} & 0
\end{bmatrix}$\\[0.08em]
\If{$\mathrm{cols}(\Qm{i}')>r$}
\State $\bm{W}_i',\,\bm{\Lambda}_i'\gets\Call{eigh}{\GYm{i}'}$
\State $\tilde{\bm{W}}_i'\gets\bm{W}_i'[:,\,:r]$
\State $\Qm{i}\gets\Qm{i}'\tilde{\bm{W}}_i'$
\State $\GXm{i}''\gets\tilde{\bm{W}}_i'^{\top}\GXm{i}'\tilde{\bm{W}}_i'$
\State $\Cm{i}''\gets\tilde{\bm{W}}_i'^{\top}\Cm{i}'\tilde{\bm{W}}_i'$
\State $\GYm{i}''\gets\bm{\Lambda}[:r,\,:r]$
\Else
\State $\Qm{i}\gets\Qm{i}'$
\EndIf
\State $\texttt{new\_direction}\gets\texttt{true}$
\EndIf
\State $\ytm{i}\gets\Qm{i}\T\ym{i}$
\State $\xtm{i}\gets\begin{cases}
\Qm{i}\T\xm{i},& \texttt{new\_direction}=\texttt{true}\\
\ytm{-},& \text{otherwise}
\end{cases}$
\State $\GXm{i}\gets\GXm{i}''+\xtm{i}\,\xtm{i}\T,\quad\GYm{i}\gets\GYm{i}''+\ytm{i}\,\ytm{i}\T$
\State $\Cm{i}\gets\Cm{i}''+\ytm{i}\,\xtm{i}\T$
\State $\ytm{-}\gets\ytm{i}$
\State $\texttt{new\_direction}\gets\texttt{false}$
\State \Return $\Qm{i},\,\GXm{i},\,\Cm{i}$
\EndProcedure
\end{algorithmic}
\end{algorithm}

In the sDMD method, the emergence of a new dynamic mode that is not already spanned by either of the existing bases requires consecutive updates to both $\Qm{\Xm{i}}$ and $\Qm{\Ym{i}}$, resulting in nearly identical computations for the same unseen direction at consecutive time steps. This redundancy arises from the maintenance of two separate bases that represent the same underlying dynamical process. Conversely, the proposed esDMD method consolidates these updates into a single basis, ensuring each new direction is accounted for only once. Importantly, maintaining two separate bases leads to consecutive eigendecompositions once the bases reach the maximum allowable number of columns, significantly increasing computational overhead. In contrast, the proposed method streamlines this process, enhancing computational efficiency without compromising accuracy.

\begin{algorithm}[!t]
\caption{Compute linear operator}
\label{alg:linear-operator}
\begin{algorithmic}
\Procedure{DynamicSpectrum}{$(\xm{i},\,\ym{i})$}
\State $\Qm{i},\,\GXm{i},\,\Cm{i}\gets\Call{UpdateBasis}{(\xm{i},\ym{i}),\,i}$
\State $\St_i\gets\Cm{i}\,\GXm{i}^+$
\State $\bm{V}_i,\,\bm{\Lambda}_i\gets\Call{eig}{\St_i}$
\State $\bm{\Phi}_i\gets\Qm{i}\,\bm{V}_i$
\State \Return $\bm{\Phi}_i,\,\mathrm{diag}(\bm{\Lambda}_i)$
\EndProcedure
\end{algorithmic}
\end{algorithm}

Although implementation details may vary by application, a direct comparison of naive reference implementations reveals that our method requires approximately 350 static bytecode instructions, whereas sDMD requires around 515. These counts correspond to the total number of static instructions within the basis update routines for each algorithm. The actual number of instructions executed per iteration may vary due to control flow and branching; however, the overall computational footprint of our method is demonstrably lower, and it could be further reduced through deliberate code-level optimizations.

\section{Numerical experiments} \label{sec:exp}
The proposed algorithm\footnote{Code available at \url{https://github.com/xakalex/esdmd}} was tested on (i) an oscillatory system generated from a sum of scaled sinusoids---the same system used in \cite{Hemati2014}, and (ii) a Kuramoto oscillator network. The results indicate that the proposed esDMD algorithm performs comparably to sDMD in capturing the dominant dynamic modes, as benchmarked against the batch-processing DMD. Both systems used for comparison are configured with $n = 100$ state variables, simulated for 10 seconds, and sampled uniformly at $f_s=120$Hz (equivalently, $\Delta t=0.0083$), yielding a total of $m=1200$ time steps. The maximum allowable rank for the streaming algorithms is set to $r=10$. The following two metrics are used for comparison:
\begin{enumerate}
    \item \emph{Spectrum}. The eigenvalues and their corresponding normalized frequencies obtained from the streaming algorithms at the final time step are compared with those calculated using the original batch-processing DMD.
    \item The \emph{execution time} of the update routines for the streaming algorithms is compared over time steps.
\end{enumerate}

\subsection{Oscillatory system}
This is a nonlinear system formed by scaled sums of sinusoids, given as
\begin{align}
\begin{split}
    \bm{y}(t)={}&\bm{v}_1\sin(2\pi f_1t)+\bm{v}_2\cos(2\pi f_2t)\\
    &+\bm{v}_3\sin(2\pi f_1t)+\bm{v}_4\cos(2\pi f_2t) 
\end{split}\label{eq.12}
\end{align}
where $\bm{v}_1,\ldots,\bm{v}_4$ are tunable vectors. In this work, they are initialized randomly following a Normal distribution $\mathcal{N}(0,1)$.
\begin{figure}
    \centering
    \subfloat[]{
    \includegraphics[width=0.9\linewidth]{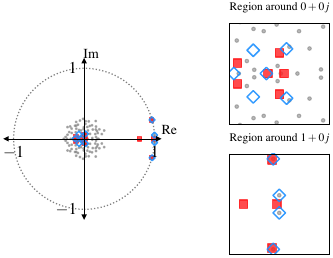}
    \label{fig:hemati-evals-sdmd}
    }
    
    \subfloat[]{
    \includegraphics[width=0.5\linewidth]{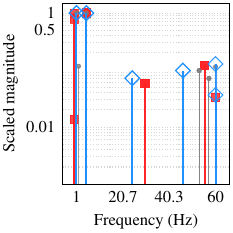}
    \label{fig:hemati-modes-sdmd}
    }
    \subfloat[]{
    \includegraphics[width=0.5\linewidth]{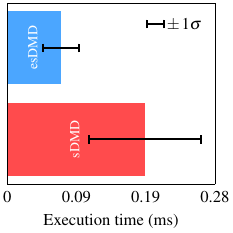}
    \label{fig:hemati-exectimes}
    }
    \caption{Comparing results obtained with DMD (grey), sDMD (red), and the proposed esDMD (blue) for the nonlinear oscillatory system \eqref{eq.12}. (a) complex discrete eigenvalues; (b) normalized mode frequencies (only the first $r$ dominant modes shown), and (c) execution time.}
    \label{fig:hemati-comparison}
\end{figure}

\subsection{Kuramoto model}
Synchronization phenomena appear in many natural and man-made systems, ranging from the synchronized flashing of fireflies to the behavior of Josephson junction arrays. The Kuramoto model is a standard framework for studying these synchronization phenomena. It describes a system of $n$ coupled oscillators with arbitrary natural frequencies, characterized by sinusoidal coupling that depends on their phase differences. Under appropriate conditions, this simple model demonstrates the emergence of phase synchronization among oscillators over time. This work uses a damped Kuramoto model adapted from \cite{kuramotoPy}, consisting of $n$ oscillators governed by:
\begin{equation}
    (1+\gamma)\,\dot{\theta}_i=\omega_i+K\sum_{j=1}^n A_{ij}\sin(\theta_j-\theta_i),\quad i=1,\ldots,n
    \label{eq:kuramoto-ode}
\end{equation}
where $\theta_i$ represents the phase of the $i$th oscillator, $\omega_i$ is its natural frequency, and $K$ is the coupling strength. The adjacency matrix $\bm{A}\in\{0,1\}^{n\times n}$ encodes the coupling topology, specifying whether any pair of oscillators interact ($A_{ij}=1$) or not ($A_{ij}=0)$. $\gamma\in[0,1]$ is the damping factor that modulates the oscillator's response, affecting the synchronization dynamics.

To evaluate the algorithms, we use the sine of the solution trajectories generated by \eqref{eq:kuramoto-ode} rather than the trajectories themselves. This transformation introduces bounded, oscillatory dynamics and can exhibit exponential convergence when the natural frequencies are close to zero. In our setup, the natural frequencies $\omega_{i}$ are drawn from the uniform distribution $\mathcal{U}[2.5, 3]$, and the damping is set to $\gamma = 0.9$.

\begin{figure}
    \centering
    \subfloat[]{
    \includegraphics[width=0.9\linewidth]{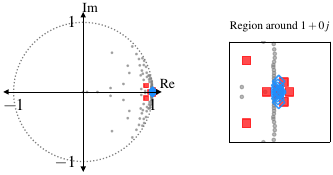}
    \label{fig:kuramoto-evals-sdmd}
    }
    
    \subfloat[]{
    \includegraphics[width=0.5\linewidth]{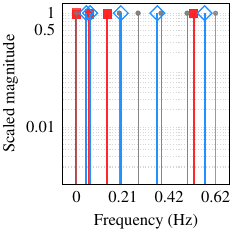}
    \label{fig:kuramoto-modes-sdmd}
    }
    \subfloat[]{
    \includegraphics[width=0.5\linewidth]{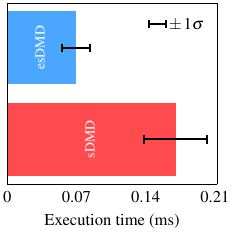}
    \label{fig:kuramoto-exectimes}
    }
    \caption{Comparing results obtained with DMD (grey), sDMD (red), and the proposed esDMD (blue) for the Kuramoto system \eqref{eq:kuramoto-ode} with one hundred oscillators. (a) complex discrete eigenvalues; (b) normalized mode frequencies (only the first $r$ dominant modes shown), and (c) execution time.}
    \label{fig:kuramoto-comparison}
\end{figure}

\subsection{Metrics comparison}
Fig. \ref{fig:hemati-comparison} and \ref{fig:kuramoto-comparison} illustrate the metrics described in section \ref{sec:exp} for the oscillatory system and the Kuramoto model, respectively. For the oscillatory system, Fig. \ref{fig:hemati-evals-sdmd} compares the complex eigenvalues computed by sDMD (red square markers) and the proposed esDMD (blue diamond markers) against the benchmark batch-processing DMD (gray circular markers). Fig. \ref{fig:hemati-modes-sdmd} displays the normalized dynamic mode frequencies in the range $[0,1]$ for both streaming methods (red and blue stems, respectively), again compared to DMD (gray stems). In all cases, both streaming variants successfully capture the dominant $r$ spectral components. Fig. \ref{fig:hemati-exectimes} presents the average execution time per iteration in milliseconds for esDMD (blue bar) and sDMD (red bar). Each horizontal bar marks the mean over all time steps, and the whiskers denote $\pm1$ standard deviation. The shorter and less variable bar for esDMD highlights its computational advantage over standard sDMD.

Analogous results for the Kuramoto model are shown in Fig. \ref{fig:kuramoto-evals-sdmd}, which compares the complex eigenvalues produced by the two streaming methods with those from batch-processing DMD. Comparisons of modal frequencies are shown in Fig. \ref{fig:kuramoto-modes-sdmd}. As with the oscillatory system, both streaming approaches accurately track the leading spectral features. The corresponding execution time comparison is provided in Fig. \ref{fig:kuramoto-exectimes}. Together, these results demonstrate that our proposed approach retains the fidelity of sDMD while reducing computational overhead, enhancing its practicality for streaming applications.

\section{Conclusions and future work} \label{sec:conc}
We introduced esDMD, a reformulation of sDMD that reduces computational redundancy by maintaining only a single orthonormal basis. This simplification not only streamlines the algorithmic workflow but also preserves the accuracy of the computed dynamic modes and eigenvalues. Unlike sDMD, which maintains two separate bases for $\bm{x}$ and $\bm{y}$ snapshots, the proposed esDMD constructs an initial orthonormal basis for the joint column space of the first snapshot pair and incrementally updates this single basis over time. The remaining steps, including the proper orthogonal decomposition compression step, closely follow the original sDMD method \cite{Hemati2014}. We evaluated the accuracy of our method by comparing it to sDMD and batch-processing DMD as benchmarks. This comparison included metrics such as complex eigenvalues, normalized modal frequencies, and the per-iteration execution time. We used two representative examples for testing: a nonlinear time-varying oscillatory system and a damped Kuramoto oscillator model. The results demonstrate that the proposed esDMD method is as effective as sDMD in capturing the dominant dynamics, but it achieves this in a fraction of the time required by sDMD.

The proposed esDMD opens the door for further enhancements and optimizations. For instance, the Gram-Schmidt orthonormalization could be replaced by more efficient techniques, such as modern variants of projection approximation subspace tracking \cite{Yang1995}. This change could greatly reduce computational overhead. The linear operator generated with esDMD could improve the prediction stage of the Koopman Kalman filter \cite{Netto2018} by introducing adaptability to systems characterized by time-varying dynamics, such as electric power grids.

\section*{Acknowledgments}
The authors are grateful to Prof. Yoshihiko Susuki from Kyoto University for his insightful comments on this work.

\bibliographystyle{IEEEtran}
\bibliography{lib}

\end{document}